\documentclass[creativecommons]{eptcs}
\providecommand{\event}{TERMGRAPH 2014}

\usepackage{nets,diagrams,graphicx,epic,pict2e}
\usepackage{multirow,breakurl}
\usepackage{amssymb,amsmath}
\usepackage{alltt,xspace}    
\usepackage{eclbkbox}
\usepackage{wrapfig}
\usepackage{multicol}

\newcommand     \lra             {\longrightarrow}
\newcommand     \ar             {{\sf ar}}
\newcommand     \myvec          {\mathaccent "017E}
\newcommand     \EQ             {\mathbin{=}}
\newcommand     \Names          {\textsf{Name}}

\newcommand     \conf[2]        {\langle{#1} \mid {#2} \rangle}
\newcommand     \confSimple[2]        {(\, #1 \mid  #2 \, )}
\newcommand     \IR             {{\cal R}}

\newcommand\ind[1] {\${#1}}

\def\bbbn{{\rm I\!N}} 
\newcommand     \NAT            {\bbbn}

\newtheorem{thm}{Theorem}

\newtheorem{definition}[thm]{Definition}
\newtheorem{lemma}[thm]{Lemma}
\newtheorem{example}[thm]{Example}

\newcommand\inn{\textsf{in}${}^2$\xspace}

\newcommand{\Func}[1]{\textsf{#1}}
\newtheorem{theorem}{Theorem}[section]
\newenvironment{program}{\begin{quote}\begin{alltt}\small{}}{\end{alltt}\end{quote}}
\newcommand{\Proof}{\noindent{}\textbf{Proof.} }
\newenvironment{proof}{\Proof}{}
\newcommand{\QED}{\ensuremath{\Box}}

\newcommand{\sym}[1]{{\mbox{\tt{}#1}}}
\newcommand{\normalscale}{0.55} 
\newcommand{\smallscale}{0.48} 
\newcommand{\tinyscale}{0.45} 
\newcommand{\Config}[2]{\langle \ #1 \, \mid \, #2 \ \rangle}

\newcommand{\CanonicalTo}{\Downarrow}
\newcommand{\rtoCom}{\rto_{\mathrm{com}}}
\newcommand{\rtoSub}{\rto_{\mathrm{sub}}}
\newcommand{\rtoCol}{\rto_{\mathrm{col}}}
\newcommand{\rtoInt}{\rto_{\mathrm{int}}}
\newcommand{\Sequence}[1]{{\vec{#1}}}
\newcommand{\ito}{\Rightarrow}         
\newcommand{\rto}{\rightarrow}         
\newcommand{\sto}{\longrightarrow}         

\newcommand{\amYAA}[3]{(\, #1 \mid  #2 \mid  #3 \, )}

\newcommand{\set}[1]{\mathrm{#1}}

\newcommand{\Code}{\Theta}
\newcommand{\NameSet}{\mathcal{N}}
\newcommand{\TermSet}{\mathcal{T}}
\newcommand{\Env}{\set{E}}

\newcommand{\Sequential}{\Longrightarrow}

\newcommand{\defeq}{\stackrel{\mathrm{def}}{=}}
\newcommand{\Empty}{-}

\newcommand{\delimStrS}{\mbox{``}}
\newcommand{\delimStrE}{\mbox{''}}
\newcommand{\Str}[1]{\delimStrS{}#1\delimStrE}
\newcommand{\Compile}{\Func{Compile}}

\title{An Implementation Model for Interaction Nets}
\def\titlerunning{An Implementation Model for Interaction Nets}

\author{Abubakar Hassan 
\institute{Theory and Practice of Software Ltd\\ London, UK}
\and 
Ian Mackie 
\institute{LIX, Ecole Polytechnique\\ 91128 Palaiseau Cedex, France}
\and Shinya Sato
\institute{University of Sussex \\Brighton, UK}}
\def\authorrunning{A. Hassan, I. Mackie \& S. Sato}

\date{\relax}  

\begin{document}
\maketitle
\bibliographystyle{eptcs}

\begin{abstract}
To study implementations and optimisations of interaction net systems
we propose a calculus to allow us to reason about nets, a concrete
data-structure that is in close correspondence with the calculus, and
a low-level language to create and manipulate this data
structure. These work together so that we can describe the compilation
process for interaction nets, reason about the behaviours of the
implementation, and study the efficiency and properties.
\end{abstract}

\section{Introduction}

Interaction nets \cite{LafontY:intn} offer a visual aspect to
rewriting. Analogous to term rewriting systems, a specific system is
defined from a user-defined set of nodes (cf.\ terms) and a
user-defined set of rewrite rules.  Nets are then graphs built from
the set of nodes and rules are graph transformations. Interaction nets
are therefore a specific, in fact very constrained, form of graph
rewriting.

Interaction nets have been used as a programming language, an
intermediate language, and as a target language for the compilation of
other programming languages. In all these application areas, prototype
implementations have been built to support the work, but they are
often not documented. The purpose of this paper is to take a new look
at the implementation of interaction nets. Specifically, we are
interested in documenting the implementation process, and in
particular showing a compilation of nets to a low-level language.  We
aim to build on the past experience and knowledge obtained from
building other implementations, and make a new contribution to this
investigation.  Specifically, we define: a calculus that can represent
and express results about interaction nets; a data-structure with a
low-level language that corresponds exactly to the calculus; and a
compilation of the calculus to this low-level language.

In addition to defining the above, an important aspect of this work is
the compilation of interaction rules: the ability to implement rules
efficiently will impact greatly on an implementation.  The low-level
language is close enough to machine code that we essentially get
atomic operations so that we can understand the cost of an
interaction. From a practical perspective, we get reliable, efficient
implementations of interaction nets from this work.  In this paper we
provide the foundations for this work, and point out a number of
directions that are currently being investigated.

A number of evaluators have been developed for interaction nets, and
one of the first abstract machines was given by Sousa
Pinto~\cite{PintoJS:seqcamin}. From another direction, a graphical
interpreter \inn was proposed by
Lippi~\cite{Eval:lippi} and it showed an aspect of interaction nets as
a visual programming tool.  Some evaluators have been proposed towards
efficient computation, called INET~\cite{Eval:HassanMS09} and
amineLight~\cite{Eval:HassanMS10}. Our approach is to build the
simplest implementation model for interaction nets that we believe can
be the most useful as well as providing the basis for more efficient
(including parallel) implementations in the future. We shall give some
evidence to support this claim in the current paper.
 
The next section recalls some background material, and in
Section~\ref{sec:calc} we describe the new calculus. In
Section~\ref{sec:data} we give the data-structure that corresponds to
the calculus together with a low-level language. We include also some
notes about the compilation of nets into the low-level language. In
Section~\ref{sec:discussion} we evaluate the work, and give some
directions for future work. We conclude the paper in
Section~\ref{sec:conc}.

\section{Background}\label{sec:background}
\begin{wrapfigure}{R}{2.6cm}
\vspace{-5mm}
\small
\begin{center}
\includegraphics[width=1.6cm,keepaspectratio,clip]{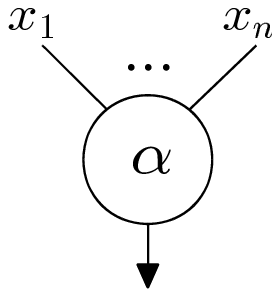}
\end{center}
\vspace{-3mm}
\caption{Agents}
\label{fig:agent}
\end{wrapfigure}
In the graphical rewriting system of interaction
nets~\cite{LafontY:intn}, we have a set $\Sigma$ of \emph{symbols},
which are names of the nodes. Each symbol has an arity $ar$ that
determines the number of \emph{auxiliary ports} that the node has. If
$ar(\alpha) = n$ for $\alpha \in \Sigma$, then $\alpha$ has $n+1$
\emph{ports:} $n$ auxiliary ports and a distinguished one called the
\emph{principal port}.  Nodes are drawn as shown in
Figure\ref{fig:agent}.  A \emph{net} built on $\Sigma$ is an
undirected graph with nodes at the vertices.  The edges of the net
connect nodes together at the ports such that there is only one edge
at every port.  A port which is not connected is called a \emph{free
  port}.
\begin{wrapfigure}[6]{R}{8.0cm}
\vspace{-5mm}
\small
\begin{center}
\includegraphics[width=7.7cm,keepaspectratio,clip]{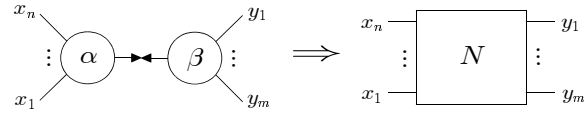}
\end{center}
\vspace{-3mm}
\caption{Interaction rules}
\label{fig:interactionrule}
\end{wrapfigure}
Two nodes $(\alpha,\beta)\in \Sigma\times\Sigma$ connected via their
principal ports form an \emph{active pair}, which is the interaction
nets analogue of a redex.  A rule $((\alpha, \beta) \ito N)$ replaces
the pair $(\alpha, \beta)$ by the net $N$.  All the free ports are
preserved during reduction, and there is at most one rule for each
pair of agents.  The diagram in Figure~\ref{fig:interactionrule}
illustrates the idea, where $N$ is any net built from $\Sigma$.  The
most powerful property of this graph rewriting system is that it is
one-step confluent---all reduction sequences are permutation
equivalent.

There are many possible data-structures that can be used to represent
interaction nets, for example single or double linked graphs.  Agents
have exactly one principal port, so we can use this to our advantage
and represent nets as a collection of trees, and use a simple, single
linked structure.  Computationally the cost of evaluating a net is
mostly due to the cost of the rewrite step: building the right-hand
side of the rule; connecting the new net to the old one; freeing up
the left-hand side of the rule (agents $\alpha$ and $\beta$ in the
diagram in Figure~\ref{fig:interactionrule}).

The goal of this paper is to build an implementation model, together
with a calculus, that explains this process and allows the study of
the cost of computation. Specifically, we provide a framework where we
can focus on the compilation of rules where we can perform the above
steps in the fewest instructions.

\section{Calculus}\label{sec:calc}

It is possible to reason about the graphical representation of nets,
but it is convenient to have a textual calculus for compact
representation. There are several calculi in the literature, and here
we review the \emph{Lightweight calculus}~\cite{Eval:HassanMS10},
which is a refined version of \cite{MackieIC:calin}.

\begin{description}
\item[Agents:] Let $\Sigma$ be a set of symbols, ranged over by
  $\alpha,\beta,\ldots$, each with a given \emph{arity} $\ar : \Sigma
  \to \NAT$. An occurrence of a symbol is called an \emph{agent}, and
  the arity is the number of auxiliary ports.
  
\item[Names:] Let $\NameSet$ be a set of names, ranged over by $x,y,z$, etc.
  $N$ and $\Sigma$ are assumed disjoint. Names correspond to wires in the graph system.
  
\item[Terms:] A term is built on $\Sigma$ and $\NameSet$ by the
  grammar: $t\; \mathrm{::=}\; x \mid \alpha(t_1,\ldots,t_n)$, \,
  where $x\in \NameSet$, $\alpha \in \Sigma$, $\ar(\alpha) = n$ and
  $t_1,\ldots,t_n$ are terms, with the restriction that each name can
  appear at most twice.  If $n = 0$, then we omit the parentheses.  If
  a name occurs twice in a term, we say that it is \emph{bound},
  otherwise it is \emph{free}.  We write $s,t,u$ to range over terms,
  and $\vec{s}, \vec{t}, \vec{u}$ to range over sequences of terms.  A
  term of the form $\alpha(t_1,\ldots,t_n)$ can be seen as a tree with
  the principal port of $\alpha$ at the root, and the terms
  $t_1,\ldots,t_n$ are the subtrees connected to the auxiliary ports
  of $\alpha$. The term $\ind{t}$ represents an indirection node which
  is created by reduction, and is not normally part of an initial
  term.

\item[Equations:] If $t$, $u$ are terms, then the unordered pair
  $t\EQ u$ is an \emph{equation}. $\Delta$ ranges over
  multisets of equations.

\item[Rules:] Rules are pairs of terms written: $\alpha(x_1, \ldots,
  x_n) \EQ \beta(y_1, \ldots, y_m) \ito \Delta$, where $(\alpha,\beta)
  \in \Sigma \times \Sigma$ is the active pair, and $\Delta$ is the
  right-hand side of the rule. All names occur exactly twice in a
  rule, and there is at most one rule for each pair of agents.  We
  call the free names $x_1,\ldots, x_n, y_1,\ldots,y_m$ in $\Delta$
  \emph{parameters} and write it as $\Delta(x_1,\ldots, x_n, y_1,
  \ldots, y_m)$.  All other names are bound.

\item[Configurations:] A \emph{configuration} is a pair
  $(\IR,\conf{\myvec{t}}{\Delta})$, where $\IR$ is a set of rules,
  $\myvec{t}$ a sequence of terms, and $\Delta$ a multiset of
  equations.  Each variable occurs at most twice in a configuration,
  and we extend the nomenclature of free and bound names from terms.
  The rules set $\IR$ contains at most one rule between any pair of
  agents, and it is closed under symmetry, thus if
  $\alpha(\vec{x})=\beta(\vec{y}) \in \IR$ then
  $\beta(\vec{y})=\alpha(\vec{x}) \in \IR$.  We use $C,C'$ to range
  over configurations.  We call $\myvec{t}$ the \emph{head} and
  $\Delta$ the \emph{body} of a configuration.

\end{description}

\begin{definition}[Names in terms]
  The set $\Names(t)$ of names of a term $t$ is defined in the
  following way, which extends to sequences of terms, equations,
  sequences of equations, and rules in the obvious way.

\begin{small}
$\Names(x)  =  \{ x\}, \quad
\Names(\alpha(t_1,\ldots,t_n)) = \Names(t_1)\cup\cdots\cup\Names(t_n), \quad
\Names(\ind{t}) =  \Names(t).
$
\end{small}
\end{definition}

  The notation $t[u/x]$ denotes a substitution that replaces the
  free occurrence of $x$ by the term $u$ in $t$. 
  This extends to  equations
  and configurations
  in the obvious way.

\begin{definition}[Instance of a rule]
  If $r$ is a rule $\alpha(x_1,\ldots,x_n) \EQ \beta(y_1,\ldots,y_m)
  \ito \Delta$, then $\widehat{\Delta}$ denotes a new generic
  \emph{instance} of $r$, that is, a copy of $\Delta$ where we
  introduce a new set of bound names so that those new names do not
  overlap with others already exists, but leave the free names
  (parameters) unchanged.  Example: if $\Delta$ is
  $\alpha(x,x)=\beta(a)$, then $\widehat{\Delta}$ is
  $\alpha(y,y)=\beta(a)$, where $y$ is a fresh name.
\end{definition}

The configuration $(\IR, \conf{\myvec{t}}{\Delta})$ represents a net
that we evaluate using $\IR$; $\Delta$ gives the set of active pairs
and the renamings of the net.  
We write $\conf{\myvec{t}}{\Delta}$ without $\IR$ when there is no ambiguity.
The roots of the terms in the head of
the configuration and the free names correspond to ports in the
interface of the net.  We work modulo $\alpha$-equivalence for bound
names.  The computation rules are defined below, and we use 
$\rto$ instead of $\rtoCom, \rtoSub, \rtoCol, \rtoInt$ when there is
no ambiguity.

{\small
\begin{description}
\item[Communication:] $\Config{\myvec{u}}{\Delta,x=t,x=s} \rtoCom
	\Config{\myvec{u}}{\Delta,t=s}$.
\item[Substitution:] 
	$\Config{\myvec{u}}{\Delta,\beta(\myvec{t})=u, x=s} \rtoSub
	\Config{\myvec{u}}{\Delta,\beta(\myvec{t})[s/x]=u}$ \ where $\beta \in \Sigma$ and $x$ occurs in $\myvec{t}$.

\item[Collect:] 
	$\Config{\myvec{u}}{\Delta,x=s} \rtoCol
		\Config{\myvec{u}[s/x]}{\Delta}$ \ where $x$ occurs in $\myvec{u}$.
\item[Interaction:]
	$\Config{\myvec{u}}{\Delta,\alpha(t_1,\dots,t_n)=\beta(s_1,\ldots,s_m)}
\rtoInt
	\Config{\myvec{u}}{\Delta,\widehat{\Delta_r}[t_1/x_1,\ldots,t_n/x_n,s_1/y_1,\ldots,s_m/y_m]}$ 

where $\alpha(x_1,\ldots,x_n)= \beta(y_1,\ldots,y_n) \ito \Delta_r \in \IR$.
\end{description}
}

\begin{figure}[t]
\begin{center}
\includegraphics[scale=\smallscale,keepaspectratio,clip]{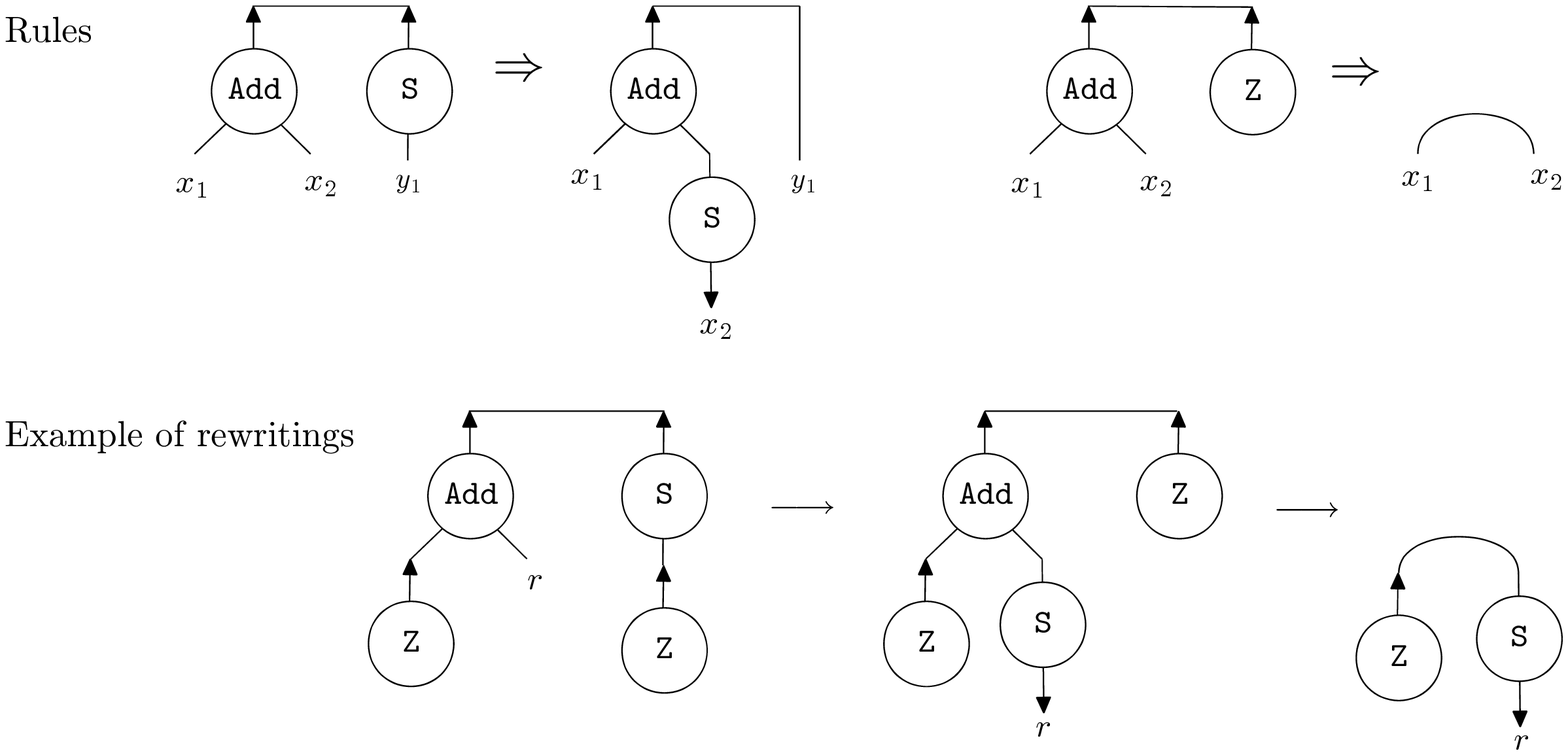}
\end{center}
\caption{An example of rules and rewritings of interaction nets}
\label{Fig:example_in}
\end{figure}

\begin{example}
Rules in Figure~\ref{Fig:example_in} can be represented as follows:

$\small
\sym{Add}(x_1, x_2)=\sym{S}(y) \ito \sym{Add}(x_1, w)=y, x_2=\sym{S}(w),
\qquad
\sym{Add}(x_1, x_2)=\sym{Z}  \ito  x_1=x_2.$

The net in Figure~\ref{Fig:example_in} is represented as $\conf{r}{\sym{Add}(\sym{Z},r)=\sym{S}(\sym{Z})}$, and it is performed:

$\small\begin{array}{lcl}
 \Config{r}{\sym{Add}(\sym{Z}, r)=\sym{S}(\sym{Z})}  
&\rtoInt &\Config{r}{\sym{Add}(\texttt{Z}, w')=\sym{Z}, r=\sym{S}(w')}
\rtoCol \Config{\sym{S}(w')}{\sym{Add}(\sym{Z}, w')=\sym{Z}}\\
&\rtoInt& \Config{\sym{S}(w')}{\sym{Z}=w'} \rtoCol \Config{\sym{S}(\sym{Z})}{}.
\end{array}$
\end{example}

We define $C_1 \CanonicalTo C_2$ by $C_1 \rto^* C_2$ where $C_2$ is in normal form. 
The following theorem shows that all Interaction rules can be performed without
using Substitution and Collect~\cite{Eval:HassanMS10}:
\begin{theorem}\label{theorem_s}
If $C_1 \CanonicalTo C_2$, 
then there is a configuration $C$ such that
$C_1 \rto^* C \rtoSub^* \rtoCol^* C_2$ and
$C_1$ is reduced to $C$ 
by applying only {\rm{}Communication} rule and 
{\rm{}Interaction} rule.
\QED
\end{theorem}

\noindent
We define $C_1 \CanonicalTo_{\rm{}ic} C_2$ by $C_1 \rto^* C_2$ where $C_2$ is a $\{\rtoInt,
\rtoCom\}-$normal form.  Because all critical pairs that are generated
by $\rtoInt$ and $\rtoCom$ are confluent, the determinacy property
holds~\cite{Eval:HassanMS10}:

\begin{theorem}[Determinacy]\label{theorem_determinacy1}
Let $C_1 \CanonicalTo C_2$. When 
there are configurations $C',C''$ 
such that $C_1 \CanonicalTo_{\rm{}ic} C'$ and
$C_1 \CanonicalTo_{\rm{}ic} C''$, then $C'$ is equivalent to $C''$. \QED
\end{theorem}

\subsection{Simpler calculus}
In Lightweight calculus, equations are defined as unordered pairs and
configurations use multisets of equations.  Here, in order to
facilitate the correspondence between the calculus and an
implementation model that has a code stack and an environment such as
SECD-machine~\cite{Landin64}, we introduce another refined calculus of
Lightweight one, called \emph{Simpler calculus}, by changing the
definition of equations into ordered pairs and in configurations
multisets of equations into sequences of ones.

\begin{description}
\item[Terms:] A term is built on $\Sigma$ and $\NameSet$ by the grammar: 
$t\; \mathrm{::=}\; x \mid \alpha(t_1,\ldots,t_n) \mid \ind{t} $.
Intuitively, $\ind{t}$ corresponds to a variable bounded with $t$ (or a state such that an environment captures $t$).

\item[Equations:] If $t$ and $u$ are terms, then the ordered pair
  $t\EQ u$ is an \emph{equation}. $\Theta$ will be
  used to range over sequences of equations.  

\item[Rules:] Rules are pairs of terms written:
$\alpha(x_1, \ldots, x_n) \EQ \beta(y_1, \ldots, y_m) \ito \Theta$.

\item[Configurations:]
A \emph{configuration} is a pair $(\IR,
\confSimple{\myvec{t}}{\Theta})$, where $\Theta$ is a sequence of
equations. 
The rules set $\IR$ contains at most one rule between any pair of agents, and it is closed under symmetry.
We use $S,S'$ to range over configurations.  
\end{description}

\begin{definition}[Computation Rules] 
  The operational behaviour of the system is given by the following:

\begin{small}
\begin{tabular}{ll}
\bigskip
\textbf{Var1:} $\confSimple{\vec{u}}{\Theta, x=t} \lra 
  \confSimple{\vec{u}}{\Theta}[\ind{t}/x]$ \  where $t \not= \ind{s}$.
&\quad
\textbf{Var2:} $\confSimple{\vec{u}}{\Theta, t=x} \lra 
  \confSimple{\vec{u}}{\Theta}[\ind{t}/x]$ \  where $t \not= \ind{s}$. \\

\smallskip\textbf{Indirection1:} $\confSimple{\vec{u}}{\Theta, \ind{t} \EQ s} 
  \lra \confSimple{\vec{u}}{\Theta, t \EQ s}$.
&
\quad
\textbf{Indirection2:} $\confSimple{\vec{u}}{\Theta, t \EQ \ind{s}} 
  \lra \confSimple{\vec{u}}{\Theta, t \EQ s}$.\\
& \\

\multicolumn{2}{l}{\textbf{Interaction:} $\alpha(x_1,\ldots,x_n) =
  \beta(y_1,\ldots,y_m) \ito \Theta_r \in \IR$, then}\\

\multicolumn{2}{l}{
\bigskip\quad
$\confSimple{\vec{u}}
{\Theta, \alpha(t_1,\ldots,t_n) = \beta(s_1,\ldots,s_m)} \lra 
\confSimple{\vec{u}}
{\Theta, \widehat{\Theta_r}[t_1/x_1,\ldots, t_n/x_n, s_1/y_1, \ldots, s_m/y_m]}.$}
\end{tabular}
\end{small}
\end{definition}

To remove indirection terms, we introduce an operation \Func{remInd}
(it is extended to sequences of terms and configurations in the obvious way):

\noindent
\begin{small}
$\begin{array}{lcl}
\Func{remInd}(x)  \defeq  x,&
\Func{remInd}(\ind{t})  \defeq  \Func{remInd}(t),&
\Func{remInd}(\alpha(t_1,\ldots,t_n))  \defeq  \alpha(\Func{remInd}(t_1),\ldots,\Func{remInd}(t_n)).
\end{array}
$
\end{small}

\begin{example}
We show the computation of the configuration
$\confSimple{r}{\sym{Add}(r,\sym{Z})=\sym{S}(\sym{Z})}$ in
Figure~\ref{Fig:example_in}:

\begin{small}
$\confSimple{r}{\sym{Add}(\sym{Z},r)=\sym{S}(\sym{Z})}
\sto \confSimple{r}{\sym{Add}(\sym{Z},x)=\sym{Z}, r=\sym{S}(x)}
\sto \confSimple{\ind{\sym{S}(x)}}{\sym{Add}(\sym{Z},x)=\sym{Z}}$

\quad
$\sto \confSimple{\ind{\sym{S}(x)}}{\sym{Z}=x}
\sto \confSimple{\ind{\sym{S}(\ind{\sym{Z}})}}{}.$

$\Func{remInd} \confSimple{\ind{\sym{S}(\ind{\sym{Z}})}}{}
 = \confSimple{\sym{S}(\sym{Z})}{}.$
\end{small}
\end{example}

These rules correspond directly to the graphical data-structure and
operations given in the next section. Indirection is introduced so
that the data-structure manipulations can be kept simple. However,
there is an overhead of dealing with indirection nodes.
Computationally the interaction rule is the most expensive: the other
rules will turn out to be implemented with a small number of
instructions or will be equivalences in the data-structure.

\subsection{Expressive power}
We compare the expressive power of Simpler and Lightweight calculi. We
define a translation of configurations from Simpler calculus into the
Lightweight one as: $ \Func{ToLight}(\confSimple{\vec{t}}{\Theta})
\defeq \conf{\Func{remInd}(\vec{t})}{\Func{remInd}(\Theta)}$.

\begin{lemma}\label{lemma:simple-simulation}
Let $S_1$ and $S_2$ be configurations such that $S_1 \lra S_2$. When it is by Indirection1 or Indirection2 rules, $\Func{ToLight}(S_1) = \Func{ToLight}(S_2)$. Otherwise, $\Func{ToLight}(S_1) \rto \Func{ToLight}(S_2)$.
\end{lemma}
\begin{proof}
In the case of Var1: $S_1 = \confSimple{\vec{u}}{\Theta, x=t} \lra 
\confSimple{\vec{u}}{\Theta}[\ind{t}/x] = S_2$. 
When we assume $\Func{ToLight}(S_1) = \conf{\vec{u'}}{\Theta', x=t'}$,
then $\Func{ToLight}(S_2) = \conf{\vec{u'}}{\Theta'}[t'/x]$, and thus
$\Func{ToLight}(S_1)\rtoCom\Func{ToLight}(S_2)$ or 
$\Func{ToLight}(S_1)\allowbreak\rtoSub\Func{ToLight}(S_2)$.
\QED
\end{proof}

Every equation is reduced by a rule in Simpler calculus, the following holds:
\begin{lemma}\label{lemma:simple-normalform}
When $S_1 \CanonicalTo \confSimple{\vec{u}}{\Theta}$, then $\Theta$ is empty. \QED
\end{lemma}

We define a translation of configurations
$\Func{ToSimple}$ from Lightweight calculus into Simpler ones:

\noindent
$\begin{array}{l}
\Func{ToSimple}(\conf{\vec{t}}{\Delta}) \defeq \confSimple{\vec{t}}{\Theta}
\ \text{where $\Theta$ is a sequence that is the result of fixing an order of the multiset $\Delta$.}
\end{array}$

\begin{theorem}
Let $C$ be a configuration in Lightweight calculus. 
When there is a configuration $S$ in Simpler one such that
$\Func{ToSimple}(C) \CanonicalTo S$,
then $C \CanonicalTo \Func{ToLight}(S)$.
\end{theorem}

\begin{proof}
Assume $S=\confSimple{\vec{u}}{\Theta}$, and then $\Theta$ is empty 
by Lemma~\ref{lemma:simple-normalform}.
Since $\Func{ToLight}\confSimple{\vec{u}}{}$ is a normal form,
$C \CanonicalTo \Func{ToLight}\confSimple{\vec{u}}{}$
by Lemma~\ref{lemma:simple-simulation}.  \QED
\end{proof}

\begin{theorem}\label{theorem:simple1}
Let $C_1$ and $C_2$ be configurations in Lightweight calculus such that $C_1 \CanonicalTo_{\rm{}ic} C_2$. 
Then there is a configuration $S$ in Simpler one such that 
$\Func{ToSimple}(C_1) \CanonicalTo S$ and 
$C_2 \CanonicalTo \Func{ToLight}(S)$.
\end{theorem}

\begin{proof}
If $\Func{ToSimple}(C_1)$ has no normal form, corresponding to an
infinite reduction sequence from $\Func{ToSimple}(C_1)$
we can construct an infinite reduction sequence starting from $C_1$ 
by Lemma~\ref{lemma:simple-simulation} since each reduction produces at most one equation such as $\ind{t}=s$ or $t=\ind{s}$.
This contradicts the assumption of this theorem. 
There is, thus, a configuration $S$ such that  
$\Func{ToSimple}(C_1) \CanonicalTo S$. 
By Theorem~\ref{theorem:simple1} $C_1 \CanonicalTo \Func{ToLight}(S)$,
and thus there is a configuration $C_3$ such that 
$C_1 \CanonicalTo_{\rm{}ic} C_3$ and $C_3 \CanonicalTo \Func{ToLight}(S)$
by Theorem~\ref{theorem_s}.
By the assumption $C_1 \CanonicalTo_{\rm{}ic} C_2$ and 
the determinacy (Theorem~\ref{theorem_determinacy1}), $C_3=C_2$. \QED
\end{proof}

We define a configuration of our abstract machine state by the
following 3-tuple $\amYAA{\set{E}}{\myvec{t}}{\Code}$, where
\begin{itemize}
\item $\set{E}$ is an environment, which is a subset of $\NameSet\times\TermSet$ ($\NameSet$ is a set of names, $\TermSet$ is the set of terms),
\item $\myvec{t}$ is an interface, which is a sequence of terms,
\item $\Code$ is a sequence of equations to operate. 
\end{itemize}

\noindent
In contrast to the SECD machine~\cite{Landin64}, the stack $S$, the
environment $E$ and the control $C$ in the machine correspond to the
term sequence $\myvec{t}$, the map $\Env$, and the equation
sequence $\Code$ in this abstract machine respectively.  There is no
element corresponding to the dump $D$ in SECD machine because, during
an execution of a rule, other rules are not called. 
To manage the environment, we define the following.

\begin{definition}[Operations for pairs]\label{Definition:pairs-and-operations}
Let $\set{P}$ be a set of pairs. 
\begin{itemize}
\item We define a map $\set{P}$ as a set of pairs:
$\set{P}(n) \defeq m  \mbox{ if $(n,m) \in \set{P}$}$, $\bot$ otherwise.

\item We use the following notations to operate maps:
\begin{small}
\begin{itemize}
\item $\set{P}(n):=\bot$ \quad as  a set \quad$(\set{P} - \{(n,m)\})$ for any $m$,
\item $\set{P}(n):=m$ \quad as a set \quad $(\set{P}[n]:=\bot) \cup \{(n,m)\})$.
\end{itemize}
\end{small}

\end{itemize}
\end{definition}

We give the semantics of the machine
as a set of the following transitional rules of the form
$\amYAA{\set{E}}{\myvec{u}}{\Code}\Sequential
\amYAA{\set{E}'}{\myvec{u}}{\Code'}$ by applying in the order from $\mathrm{A}$ to $\mathrm{C2}$:
$$\small
\begin{array}{rlcl}
\mathrm{A:} & 
\amYAA{\set{E}}{\myvec{u}}{\Theta, \alpha(\myvec{t})=\beta(\myvec{s})} 
& \Sequential & \amYAA{\set{E}}{\myvec{u}}{\Theta, \Theta_1}
\quad \text{where $\confSimple{}{\alpha(\myvec{t})=\beta(\myvec{s})} \lra \confSimple{}{\Theta_1}$}\\
\mathrm{B1:} &
\amYAA{\set{E}(x)=\bot}{\myvec{u}}{\Theta, x=t} 
& \Sequential & \amYAA{\set{E}(x):=t}{\myvec{u}}{\Theta} \\
\mathrm{B2:} & \amYAA{\set{E}(x)=\bot}{\myvec{u}}{\Theta, t=x} 
& \Sequential & \amYAA{\set{E}(x):=t}{\myvec{u}}{\Theta} \\
\mathrm{C1:} & \amYAA{\set{E}(x)=s}{\myvec{u}}{\Theta, x=t} 
& \Sequential & \amYAA{\set{E}(x):=\bot}{\myvec{u}}{\Theta, s=t}\\
\mathrm{C2:} & \amYAA{\set{E}(x)=s}{\myvec{u}}{\Theta, t=x} 
& \Sequential & \amYAA{\set{E}(x):=\bot}{\myvec{u}}{\Theta, t=s}
\end{array}$$

Intuitively, the rule `A' corresponds Interaction rule,
`B1' and `B2' correspond Var1 and Var2,
and `C1' and `C2' correspond Indirection1 and Indirection2.
To force captured terms in the environment to be replaced when the execution is finished, 
we define $\Func{Update}$ operation:
$$
\small
\begin{array}{lcl}
\Func{Update}\amYAA{\set{E}\cup\{(x,s)\}}{\myvec{u}}{\Code}
&=&\left\{
\begin{array}{ll}
\Func{Update}\amYAA{\set{E}[s/x]}{\myvec{u}[s/x]}{\Code[s/x]} & 
\mbox{(when $x$ occurs in $\set{E}, \myvec{u}$ or $\Code$)}\\
\Func{Update}\amYAA{\set{E}}{\myvec{u}}{\Code,x=s} & 
\mbox{(otherwise)}\\
\end{array}
\right.
\\
\Func{Update}\amYAA{\emptyset}{\myvec{u}}{\Code}
&=&\confSimple{\myvec{u}}{\Code}
\end{array}
$$

\begin{example}\label{Example:light-rule-adds-addz}
A configuration
$\confSimple{r}{\sym{Add}(\sym{Z},r)=\sym{S}(\sym{Z})}$ 
which represents the net 
in Figure~\ref{Fig:example_in} is performed:

\begin{small}
\noindent{}$\amYAA{\emptyset}{r}{\sym{Add}(\sym{Z},r)=\sym{S}(\sym{Z})}
\Sequential \amYAA{\emptyset}{r}{\sym{Add}(\sym{Z},x)=\sym{Z}, r=\sym{S}(x)}
\Sequential \amYAA{\{(r,\sym{S}(x))\}}{r}{\sym{Add}(\sym{Z},x)=\sym{Z}}$

\quad$\amYAA{\{(r,\sym{S}(x))\}}{r}{\sym{Z}=x}
\Sequential \amYAA{\{(r,\sym{S}(x)), (x,\sym{Z})\} }{r}{\Empty}.$

\noindent{}$\Func{Update}\amYAA{\{(r,\sym{S}(x)), (x,\sym{Z})\} }{r}{\Empty}
= \Func{Update}\amYAA{\{ (r,\sym{S}(\sym{Z})) \}}{ r }{\Empty }  
= \Func{Update}\amYAA{\emptyset}{\sym{S}(\sym{Z})}{\Empty }
= \confSimple{\sym{S}(\sym{Z})}{}.$
\end{small}
\end{example}

\section{Data-structures and language}\label{sec:data}

\begin{wrapfigure}[6]{R}{6.4cm}
\vspace{-4mm}
\includegraphics[width=6.2cm,keepaspectratio,clip]{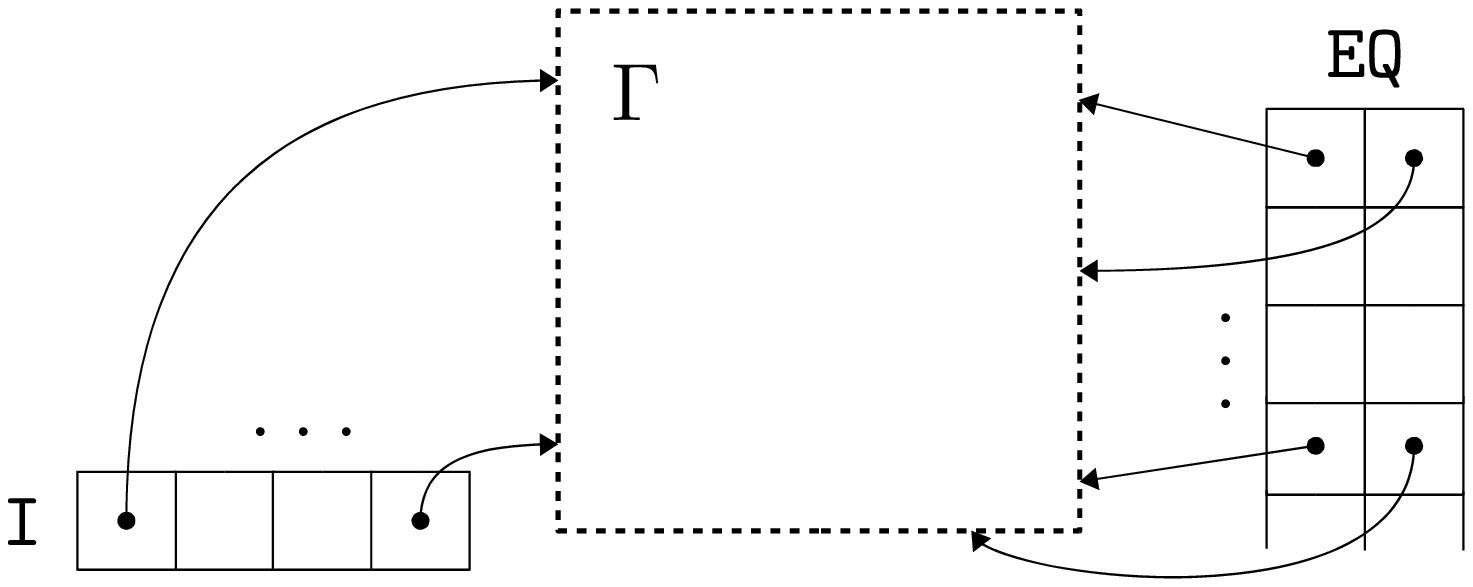}
\end{wrapfigure}
Here we give a low-level language, called LL0, which
defines a set of instructions to build and reduce a net to
normal form.  The concrete representation of a configuration can be
summarised by the diagram shown on the right, where
$\Gamma$ represents the net, $\sym{EQ}$ a stack of
equations, and $\sym{I}$ an interface.

For a net, we need two kinds of graph element: agents and names nodes.
Each of these is allocated memory in the heap.  An element, such as an
agent, may contain pointers to other elements (representing auxiliary
ports).  An agent can be coded in C as follows:
\begin{small}
\begin{program}
typedef struct Agent \{
  int id; struct Agent *port[];
\} Agent;  
\end{program}
\end{small}

\noindent
In this \sym{Agent} structure, each symbol
$\alpha_{1},\ldots,\alpha_{n}$ for agents is distinguished by a unique
\sym{id}.  The length of \sym{port} corresponds to the number of
auxiliary ports of an agent.  The stack of equations $\sym{EQ}$ is
initially empty.  Intuitively, an element of this stack can be written
using the following code fragment in C:
\begin{small}
\begin{program}
typedef struct Equation \{
  Agent *a1; Agent *a2;
\} Equation;
\end{program}
\end{small}

The interface $\sym{I}$ is a node arrays of fixed size $n$ 
as the size of the observable interface of a net 
can be pre-determined (and it is preserved during execution).
By using LL0, we encode Simpler calculus.

\paragraph{Building nets.} 
Graph elements, the stack of equations $\sym{EQ}$ and
the interface $\sym{I}$ are managed by instructions as shown 
in Figure~\ref{fig:instructions-LL0}.
The port numbers start from 1, and by using the instruction
$x\sym{[}p\sym{]} \sym{=} y$, we can assign a graph node $y$ into a
port $p>0$ of a graph node $x$. We also use the port 0 to refer to the id
of an element.  For instance, $x\sym{[0]} \sym{=} \alpha$ changes the id of
an agent node $x$ into $\alpha$.

\begin{figure}[t]
\begin{center}
{\rm\small
\begin{tabular}{lp{11.1cm}}\hline
Instruction & Description\\\hline
$\sym{\#agent } \alpha_1:p_1\sym{,} \cdots \sym{,} \alpha_n:p_n$
& 
Declare $\alpha_1,\cdots , \alpha_n$ as symbols of agents whose arity are $p_1,\cdots,p_n$.\\
$\sym{I=mkInterface(}n\sym{)}$& 
Create a fixed $n$-size interface and assign its pointer to the variable $\sym{I}$.\\
$x \, \sym{=} \, \sym{mkAgent(}\mathit{id}\sym{)}$ & 
Allocate (unused) memory for an agent node whose id is $\mathit{id}$ 
and assign it to the variable $x$.\\
$x \, \sym{=} \, \sym{mkName()}$ & 
Allocate (unused) memory for a name node, and assign it to the variable $x$.\\
%
%
$\sym{free(}x\sym{)}$ & Dispose of just an assigned allocation $x$ of a graph element (not recursively).\\
$x\sym{[}p\sym{]} \sym{=} y$ & Assign a graph element $y$ to a port $p >0$ of an agent node $x$.\\
$x\sym{[0]} \sym{=} \alpha$ & Change the id of an agent node $x$ into $\alpha$.\\
$\sym{push(}x\sym{,}y\sym{)}$ & 
Create an equation of two graph element $x, y$ in the stack of 
equations.\\
$\sym{stackFree()}$ & 
Dispose of the top element of the equation stack.\\
\hline
\end{tabular}}
\end{center}
\caption{Instructions of LL0}
\label{fig:instructions-LL0}
\end{figure}

\begin{wrapfigure}[3]{R}{2.2cm}
\vspace{-4mm}
\includegraphics[scale=\normalscale,keepaspectratio,clip]{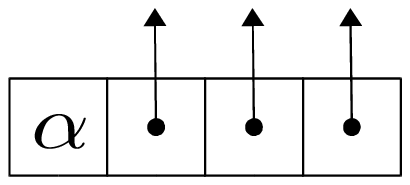}
\end{wrapfigure}
Here, we build terms in Simpler calculus.
To assign an arity to an agent, we use the following declaration:
$\sym{\#agent } \alpha_1:p_1\sym{,} \cdots \sym{,} \alpha_n:p_n$,
where $p_i$ is the arity for an agent symbol $\alpha_i$ such that
$ar(\alpha_i) = p_i$.  After this declaration, a symbol $\alpha_i$
can be represented by a unique number and an agent's arity $p_i$ can be
referred to by
$\Func{arity}(\alpha_i)=p_i$.
We draw an agent node 
$\alpha$ of arity 3 as shown the right above figure.

\begin{wrapfigure}[3]{R}{2.2cm}
\vspace{-4mm}
\includegraphics[scale=\normalscale,keepaspectratio,clip]{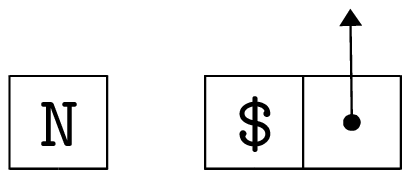}
\end{wrapfigure}
Name nodes are graph elements whose $\sym{id}$ is denoted 
by the symbol $\sym{N}$ and the arity is 0. 
We also use indirection nodes by $\sym{\$}$ and the arity is 1.
We assume that $\sym{N}$ and $\sym{\$}$ are selected from a set
that does not overlap with the set of agent symbols.
To allocate an agent node whose id is $\mathit{id}$ to a variable $x$, 
we use the following instruction:
$x\sym{=mkAgent(}\mathit{id}\sym{)}$.
A name node is allocated by $x\sym{=mkName()}$.
An assigned allocation $a$ of a graph element is disposed of (not recursively, just one node) by using 
$\sym{free(}a\sym{)}$. 

A connection between a principal port and an auxiliary port is encoded
by an assignment.  In this language, to assign a pointer of an
existing graph element $b$ to a port $p$ of another graph element $a$,
we use the following instruction: $a\sym{[}p\sym{]} \sym{=} b$.  We
note that the index of these ports start from 1.  For instance, a term
$\sym{Add(Z,r)}$ is encoded as shown below, together with the
graphical representation.
\vspace{-3mm}
\begin{center}
\begin{multicols}{3}
\small
\begin{program}
 1. #agent Z:0, Add:2
 2. aAdd=mkAgent(Add)
 3. aZ=mkAgent(Z)
 4. aAdd[1]=aZ
 5. r=mkName()
 6. aAdd[2]=r
\end{program}
\includegraphics[scale=\normalscale,keepaspectratio,clip]{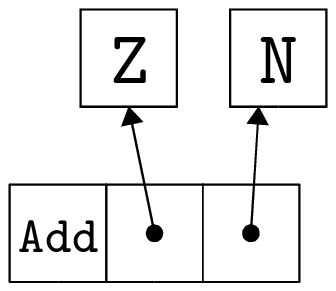}
\end{multicols}
\end{center}

\begin{wrapfigure}[2]{R}{1cm}
\vspace{-4mm}
\includegraphics[scale=\normalscale,keepaspectratio,clip]{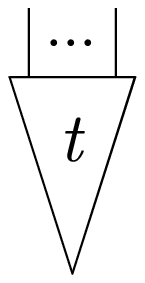}
\end{wrapfigure}
Generally, when agent nodes are connected together, 
they are trees that we represent in the following way,
where the free ports are at the top of the tree.

\begin{wrapfigure}[4]{R}{2.3cm}
\vspace{-2mm}
\includegraphics[scale=\normalscale,keepaspectratio,clip]{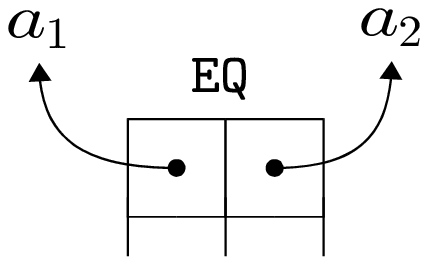}
\end{wrapfigure}
The stack of equations $\sym{EQ}$ is initially created empty.
An equation node can point to two graph elements.
To create an equation of two graph elements $a_1, a_2$ 
in the stack $\sym{EQ}$, we use the instruction:
$\sym{push(}a_1\sym{,}a_2\sym{)}$.
To pop an equation from the top of the stack $\sym{EQ}$, we use the instruction:
\texttt{stackFree()}.
We represent a connection between
principal ports by creating an equation between the two agent nodes
into the stack.  

\begin{wrapfigure}[7]{R}{5.8cm}
\vspace{-4mm}
\includegraphics[scale=\tinyscale,keepaspectratio,clip]{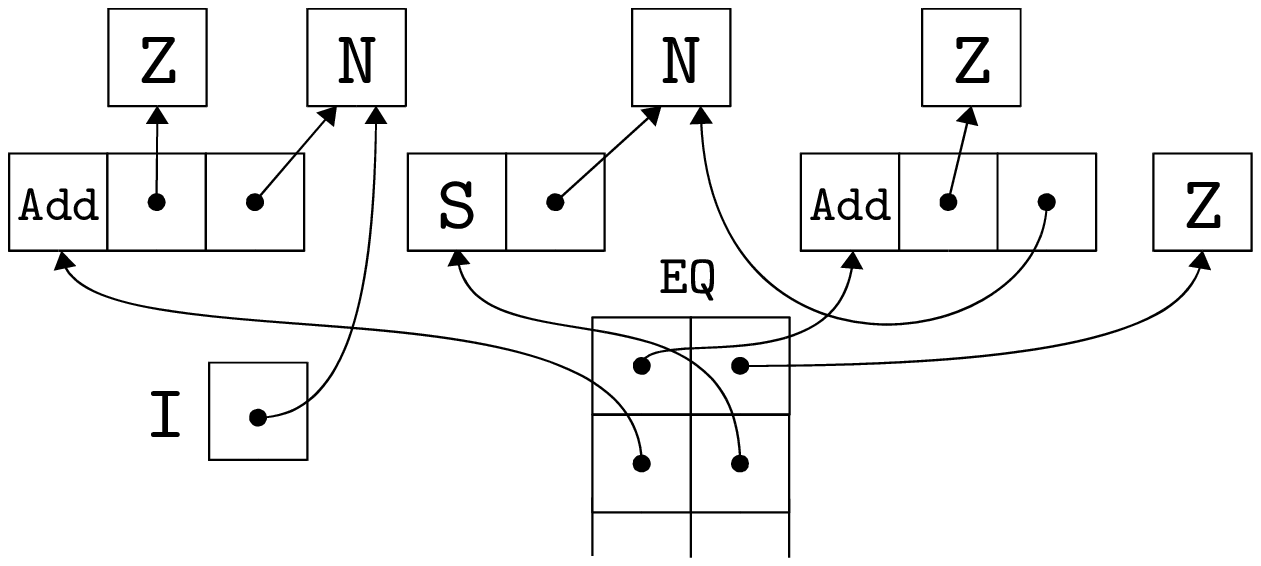}
\vspace{-7mm}
\caption{\small{}Representation of $\confSimple{r}{\sym{Add}(\sym{Z},r)=\sym{S}(w), \, \sym{Add}(\sym{Z},w)=\sym{S}(\sym{Z})}$ }
\label{example_connection_auxs_node2}
\end{wrapfigure}
Interfaces are created with the instruction:
$\sym{I=mkInterface(}n\sym{)}$.  Elements in $\sym{I}$ can be accessed
using the usual array notation $\sym{I[1]},\cdots,\sym{I[n]}$ and can
point to one graph element.  As an example, a configuration
$\confSimple{r}{\sym{Add}(\sym{Z},r)=\sym{S}(w), \,
  \sym{Add}(\sym{Z},w)=\sym{S}(\sym{Z})}$ is encoded using the
instructions below. Figure \ref{example_connection_auxs_node2} gives
the corresponding data-structure. For a connection between
two auxiliary ports, we assign one name node to two ports.

\medskip
\begin{small}
\begin{multicols}{3}
\begin{program}
 1. #agent Z:0,S:1,Add:2
 2. /*interface*/
 3. I=mkInterface(1)
 4. /*Add(Z,r)*/
 5. aAdd=mkAgent(Add)
 6. aZ=mkAgent(Z)
 7. aAdd[1]=aZ
 8. r=mkName()
 9. aAdd[2]=r
10. /*S(w)*/
11. bS=mkAgent(S)
12. w=mkName()
13. bS[1]=w
14. /*Add(Z,r)=S(w)*/
15. push(aAdd,bS)
16. /*Add(Z,w)*/
17. aAdd=mkAgent(Add)
18. aZ=mkAgent(Z)
19. aAdd[1]=aZ
20. aAdd[2]=w
21. /*S(Z)*/
22. bS=mkAgent(S)
23. bZ=mkAgent(Z)
24. bS[1]=bZ
25. /*Add(Z,w)=S(Z)*/
26. push(aAdd,bS)
27. /*interface*/
28. I[1]=r
\end{program}
\end{multicols}
\end{small}

\paragraph{Defining interaction rules.}\label{Section:how_to_define_interaction_rules}
We introduce \emph{rule procedures} to perform interaction rules.  For
an interaction rule between $\alpha(\Sequence{x})$ and
$\beta(\Sequence{y})$ we define a rule procedure  using the syntax:
$\sym{rule }\alpha\sym{ } \beta\sym{ \{} \ldots \sym{\}}$ and
we write instructions between the brackets $\sym{\{}$ and $\sym{\}}$ (rule block).
In execution, the procedures provide special variables $\sym{L},
\sym{R}$ that are pointers to the left and the right-hand side agents
of the active pair equation. Variables used in the instructions are only 
visible within the rule procedure.
Generally, these rule procedures are represented as transformations on the data-structure. 
For instance, the rule between $\sym{Add}$ and $\sym{Z}$ given by 
$\sym{Add}(x_1, x_2)=\sym{Z} \ito x_1=x_2$ is represented using the following procedure:
\begin{small}
\begin{multicols}{3}
\begin{program}
 1. rule Add Z \{ 
 2.   stackFree() 
 3.   push(L[1],L[2]) 
 4.   free(L)
 5.   free(R)
 6. \}
\end{program}
\end{multicols}
\end{small}
\noindent{}
The following illustrates transformations that will be applied by 
the rule procedure given above. 
\begin{center}
\includegraphics[scale=\smallscale,keepaspectratio,clip]{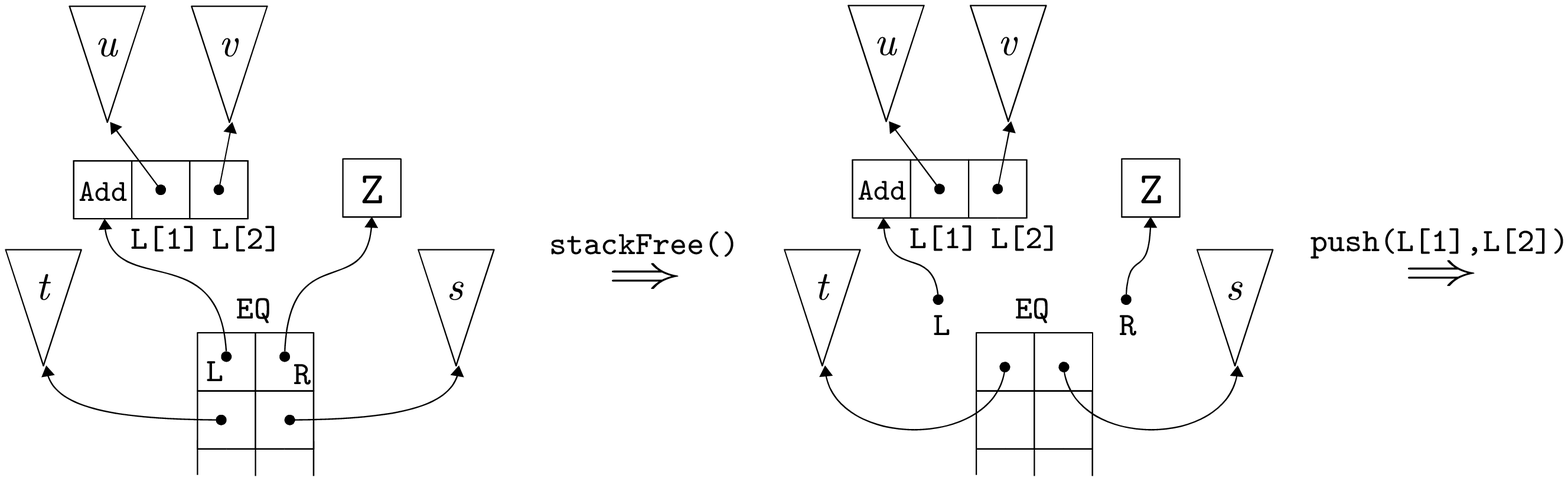}

\medskip
\includegraphics[scale=\smallscale,keepaspectratio,clip]{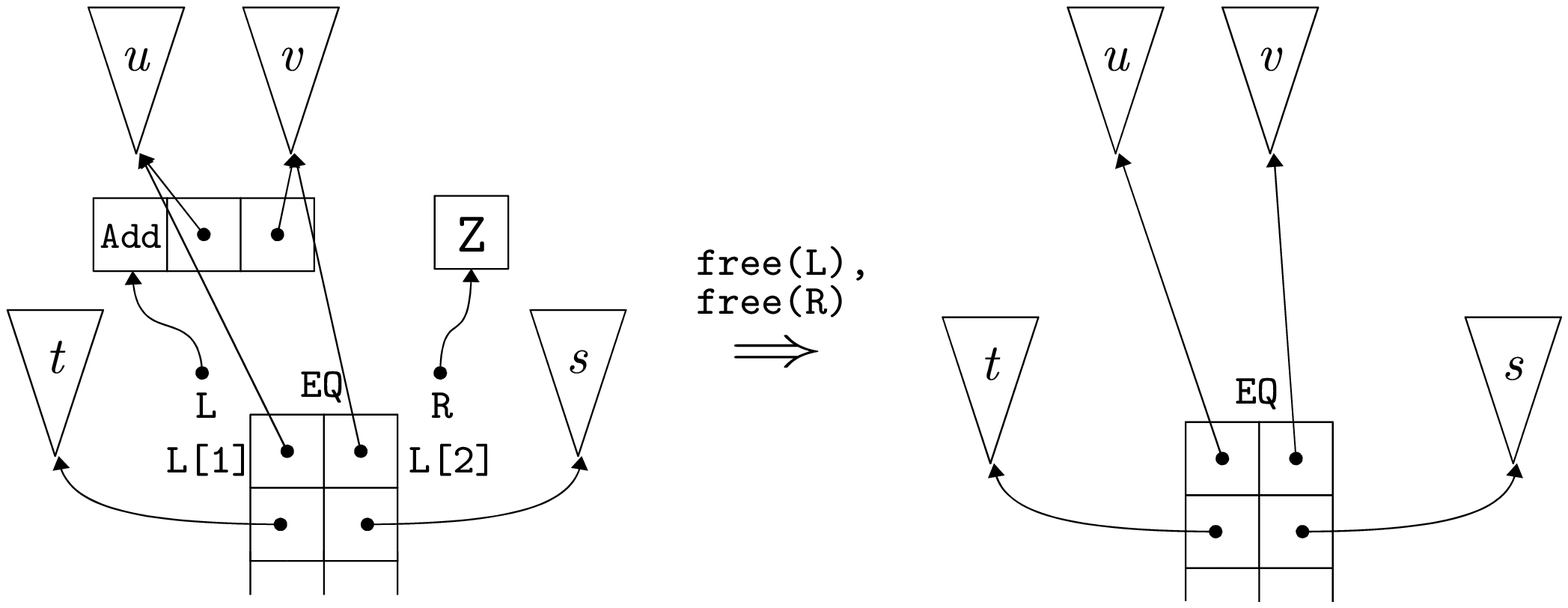}
\end{center}
In order to manage equations in the Simpler calculus, we need
mechanisms that perform rules Var1, Indirection1 and so on.
Figure~\ref{fig:computation-rules-for-name-and-ind} (a) and (b) are
instances of Var1 and Indirection1 rules to illustrate this.

\begin{figure}[t]
\quad\begin{minipage}[t]{7.9cm}
\small
(a) $x = t,u = s \sto u[\ind{t}/x] = s$ \, where $x\in \NameSet(u)$
\begin{center}
\includegraphics[scale=\tinyscale,keepaspectratio,clip]{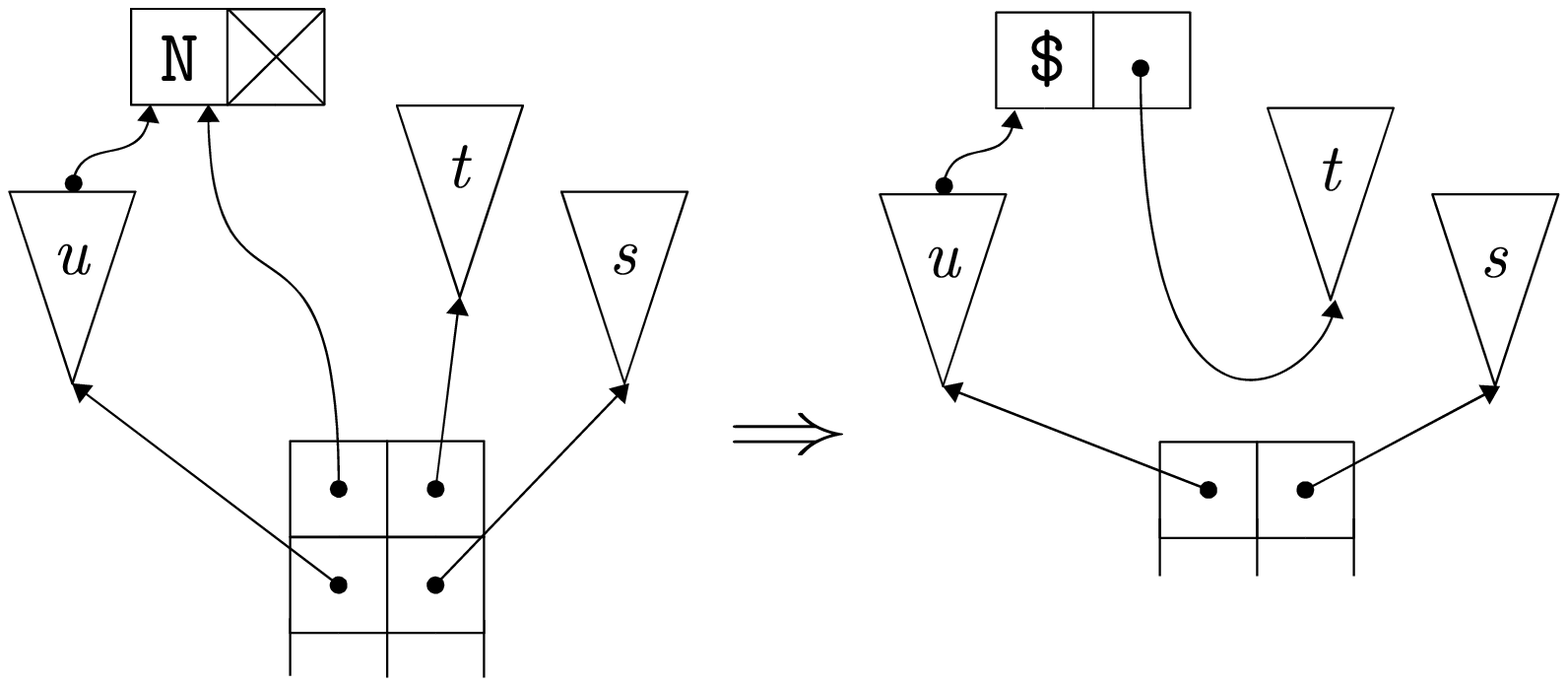}
\end{center}
\end{minipage}
\qquad\qquad
\begin{minipage}[t]{5.3cm}
\small
(b) $\ind{t} = s \sto t = s$
\begin{center}
\includegraphics[scale=\tinyscale,keepaspectratio,clip]{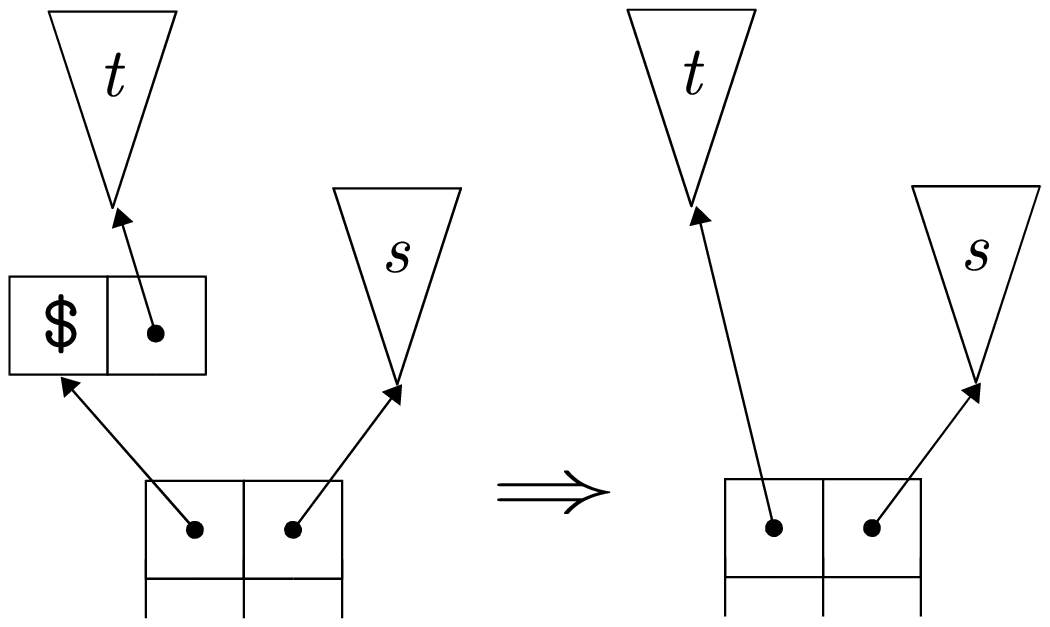}
\end{center}
\end{minipage}
\caption{Computation rules for name and indirection nodes}
\label{fig:computation-rules-for-name-and-ind}
\end{figure}

\paragraph{Compilation of Simpler calculus into LL0.}
Here we introduce a translation of Simpler calculus into LL0.
We use a set of pairs and operations for the pairs defined in
Definition~\ref{Definition:pairs-and-operations}.  We also use a
notation for strings. We write $\delimStrS$ and $\delimStrE$ as a
pair of delimiters to represent a string explicitly.  We use the
notation $\{x\}$ in a string as the result of replacing the occurrence
$\{x\}$ with its actual value.  For instance, if $x=\Str{\sym{abc}}$
and $y=89$ then $\Str{\sym{1}\{x\}\sym{2}\{y\}} =
\Str{\sym{1abc289}}$.  We use $+$ as an infix binary operation to
concatenate strings.

\begin{definition}[Compilation of terms and nets]
Here we defined our compilation schemes that will generate LL0 for a
given interaction net system.
\begin{itemize}
\item We use a subset $\set{N}$ of $\NameSet \times \set{Str}$
($\NameSet$ is a set of names)
so that a name $x \in \NameSet$ can correspond to 
a string of a variable name in a code sequence and
those correspondences can be looked up from compilation functions.
We also use two operations $\set{N}(x):=\bot$ and $\set{N}(x):=\mathrm{str}$ for $\mathrm{str} \in \set{Str}$ as defined in Definition~\ref{Definition:pairs-and-operations}.
We define a function $\Func{makeN}$ to make such a set $\set{N}$ and 
a code sequence for those names by a given name set $\{x_1,\ldots,x_n\}$. The 
function $\Func{freshStr()}$ returns a fresh name. 
\begin{small}
$$
\begin{array}{lcl}
\Func{makeN}(\{x_1,\ldots,x_n\}) &\defeq&
\begin{array}[t]{l}
\Func{makeN'}(\{x_1,\ldots,x_n\}, \emptyset);
\end{array}\\
\Func{makeN'}(\{x_1,\ldots,x_n\}, \set{N}) &\defeq&
\begin{array}[t]{l}
\Func{let} \quad\set{N}_0 = \emptyset\Func{;}\\
a_1=\Func{freshStr();}
\quad{}c_1 = \Str{\{a_1\} \sym{=mkName()}}; 
\quad\set{N}_1 = (\set{N}_0(x_1):=a_1)\Func{;}
\cdots\\
a_n=\Func{freshStr();}
\quad{}c_n = \Str{\{a_n\} \sym{=mkName()}}; 
\quad\set{N}_n = (\set{N}_{n-1}(x_n):=a_n)\Func{;}\\

\Func{in} \quad(c_1 + \cdots + c_n, \set{N}_n) \quad\Func{end;}\\
\end{array}\\
\end{array}
$$
\end{small}

\item We define a translation $\Compile_s$ from a symbol set $\Sigma$
into a code string as follows:

\begin{small}
$\begin{array}{lcl}
\Compile_s(\emptyset)
&\defeq &
\begin{array}[t]{l}
\Str{}\\
\end{array}\\
\end{array}$

$\begin{array}{lcl}
\mid \Compile_s(\{\alpha_1,\ldots,\alpha_n\})
&\defeq&
\begin{array}[t]{l}
\Str{\sym{\#agent }}+
\Str{\{\alpha_1\}\sym{:}\{ar(\alpha_1)\}}+
\cdots + 
\Str{\sym{,}\{\alpha_n\}\sym{:}\{ar(\alpha_n)\}}\Func{;}\\
\end{array}\\
\end{array}$
\end{small}

\item A translation $\Compile_t$ from a term 
into a code string is defined as follows:

\begin{small}
$\begin{array}{lcl}
\Compile_t(x,\set{N})
& \defeq &
\begin{array}[t]{l}
(\Str{}, \set{N}(x))\\
\end{array}\\
\end{array}
$

$
\begin{array}{lcl}
\mid \Compile_t(\alpha(t_1,\ldots,t_n),\set{N})
&\defeq&
\begin{array}[t]{l}
\Func{let} \quad{}a = \Func{freshStr()}; 
\quad{}c = \Str{\{a\} \sym{=mkAgent(} \{\alpha\}\sym{)}};\\
\qquad{}(c_1, \, a_1) = \Compile_t(t_1,\set{N}); 
\quad{}c_1 = c_1 + \Str{\{a\}\sym{[1]=}\{a_1\}};
\quad{}\cdots\\
\qquad{}(c_n, \, a_n) = \Compile_t(t_n,\set{N});
\quad{}c_n = c_n + \Str{\{a\}\sym{[}n\sym{]=}\{a_n\}};\\
\Func{in} \quad(c + c_1 + \cdots + c_n, \, a) \quad
\Func{end};\\
\end{array}
\end{array}
$
\end{small}

\item A translation $\Compile_i$ from an interface $u_1,\ldots,u_n$ into a code sequence is defined as follows:

\begin{small}
$\begin{array}{lcl}
\Compile_i(-,\set{N})
& \defeq &
\Str{}\\
\end{array}$

$\begin{array}{lcl}
\mid \Compile_i(u_1,\ldots,u_n, \set{N})
& \defeq &
\begin{array}[t]{l}
\Func{let}
\quad{}(c_1, \, a_1) = \Compile_t(u_1,\set{N});
\quad{}c_1 = c_1+\Str{\sym{I[1]=}\{a_1\}}\Func{;}
\quad{}\cdots\\
\qquad{}(c_n, \, a_n) = \Compile_t(u_n,\set{N});
\quad{}c_n = c_n+\Str{\sym{I[}n\sym{]=}\{a_n\}}\Func{;} \\
\Func{in} \quad{}\Str{\sym{I=mkInterface[}n\sym{]}}+c_1+\cdots+c_n \quad \Func{end;}\\
\end{array}\\

\end{array}$

\end{small}

\item A translation $\Compile_e$ from an equation 
into a code string is defined as follows:

\begin{small}
$\begin{array}{l}
\Compile_e(t=s, \set{N})
\defeq 
\begin{array}[t]{l}
\Func{let}
\quad{}(c_1,\, a_1) = \Compile_t(t,\set{N});
\quad{}(c_2,\, a_2) = \Compile_t(s,\set{N});\\
\Func{in}
\quad c_1 + c_2 + \Str{\sym{push(}\{a_1\}\sym{,}\{a_2\}\sym{)}} \quad\Func{end};\\
\end{array}\\
\end{array}$
\end{small}

\item A translation $\Compile_{es}$ from an equation sequence
into a code string is defined as follows:

\begin{small}
$\begin{array}{l}
\Compile_{es}(e_1,\ldots,e_n,\set{N})
 \defeq 
\begin{array}[t]{l}
\Compile_e(e_1,\set{N}) + \cdots + \Compile_e(e_n,\set{N})\Func{;}\\
\end{array}\\
\end{array}$
\end{small}

\item We define a translation $\Compile_c$ from a configuration
$\confSimple{\Sequence{u}}{\Theta}$ with a symbol set $\Sigma$
into a code string $c$ as follows:

\begin{small}
$\begin{array}{l}
\Compile_c(\Sigma, \, \confSimple{\Sequence{u}}{\Theta})
 \defeq 
\begin{array}[t]{l}
\Func{let}
\quad c_0 = \Compile_s(\Sigma);
\quad (c_1,\set{N}) = 
  \Func{makeN}(\Func{Name}\confSimple{\Sequence{u}}{\Theta});\\
\qquad c_2 = \Compile_{es}(\Theta,\set{N});
\quad c_3 = \Compile_i(\Sequence{u},\set{N});\\
\Func{in} \quad c_0 + c_1+c_2+c_3 \quad \Func{end};\\
\end{array}\\

\end{array}$
\end{small}

\item We write just $\Compile$ when there is no ambiguity.

\end{itemize}
\end{definition}

\begin{example}\label{Example:compilation1}
Let us take a configuration $\confSimple{r}{\sym{Add}(\sym{Z},r)=\sym{S}(\sym{Z})}$ with 
a symbol set $\{\sym{Z}, \sym{S}, \sym{Add}\}$ as an example. The compilation 
$\Compile_c(\{\sym{Z}, \sym{S}, \sym{Add}\}, \, \confSimple{r}{\sym{Add}(\sym{Z},r)=\sym{S}(\sym{Z})})$ generates the following instructions:
\begin{small}
\begin{multicols}{3}
\begin{program}
 1. \#agent Z:0,S:1,Add:2
 2. r=mkName()
 3. a1=mkAgent(Add)
 4. a2=mkAgent(Z)
 5. a1[1]=a2
 6. a1[2]=r
 7. b1=mkAgent(S)
 8. b2=mkAgent(Z)
 9. b1[1]=b2
10. push(a1,b1)
11. I=mkInterface[1]
12. I[1]=r
\end{program}
\end{multicols}
\end{small}
\end{example}

\begin{definition}[Compilation of rules]
We define a translation $\Func{Compile}_r$ from a rule into a sequence of code strings as follows:

\begin{minipage}[t]{8.5cm}
\begin{small}
$\begin{array}[t]{l}
\Compile_r(\alpha(\Sequence{x})=\beta(\Sequence{y})\ito \Theta)
 \defeq \\
\begin{array}[t]{l}
\Func{let}\\
\quad \set{N}_l = \Compile_{rn}(\Sequence{x}, \sym{L}, \emptyset);
\quad \set{N}_r = \Compile_{rn}(\Sequence{y}, \sym{R},\set{N}_l);\\
\quad (c_1,\set{N}) = \Func{makeN'}(\Func{Name}(\Theta) - \{\Sequence{x}, \Sequence{y}\}, \set{N}_r);\\
\quad{}c_2 = \Compile_{es}(\Theta,\set{N});\\
\Func{in}\\
\quad\Str{\sym{rule } \{\alpha\}\sym{ }
\{\beta\}\sym{ \{}}\\
\quad + \Str{\sym{stackFree()}}\\
\quad + c_1 + c_2\\ 
\quad + \Str{\sym{free(L)}} + \Str{\sym{free(R)}} + \Str{\sym{\}}}\\
\Func{end};\\
\end{array}
\end{array}$
\end{small}
\end{minipage}
\quad
\begin{minipage}[t]{7cm}
\begin{small}
$\begin{array}[t]{l}
\Compile_{rn}((x_1,\ldots,x_n), \, LR, \, \set{N})
\defeq\\
\begin{array}[t]{l}
\Func{let}\\
\quad{}\set{N}_0 = \set{N}\Func{;}\\
\quad{}\set{N}_1 = (\set{N}_0(x_1):=\{LR\}\sym{[1]})\Func{;}\\
\qquad{}\vdots\\
\quad{}\set{N}_n = (\set{N}_{n-1}(x_n):=\{LR\}\sym{[}n\sym{]})\Func{;}\\
\Func{in}\\
\quad \set{N}_n\\
\Func{end;}\\
\end{array}
\end{array}$
\end{small}
\end{minipage}

\end{definition}

\begin{example}\label{Example:compilation-Add-SZ}
The results of $\Compile_r(\sym{Add}(x_1,x_2)=\sym{Z}\ito x_1=x_2)$ 
and $\Compile_r(\sym{Add}(x_1,x_2) = \sym{S}(y)\ito \sym{Add}(x_1,w)=y,x_2=\sym{S}(w))$
are as follows:
\begin{small}
\begin{multicols}{3}
\begin{program}
 1. rule Add Z \{
 2.   stackFree()
 3.   push(L[1],L[2])
 4.   free(L)
 5.   free(R)
 6. \}
\end{program}

\begin{program}
 1. rule Add S \{
 2.   stackFree()
 3.   w=mkName()
 4.   a1=mkAgent(Add)
 5.   a1[1]=L[1]
 6.   a1[2]=w
 7.   push(a1,R[1])
 8.   b1=mkAgent(S)
 9.   b1[1]=w
10.   push(L[2],b1)
11.   free(L)
12.   free(R)
13. \}
\end{program}
\end{multicols}
\end{small}
\end{example}

\paragraph{Back-end of the compilation.}
Here we show how these translated codes are evaluated on the
standardised implementation model in the C language, showing the
correspondence of codes in LL0 with the C language.

\begin{itemize}
\item $\sym{\#agent } \alpha_1:p_1\sym{,} \ldots \sym{,} \alpha_n:p_n$.
For each sort of agent, we assign a unique number that is greater than 1.
We also assign 0 to the id for name nodes.
The declaration for agent symbols corresponds as follows:
\begin{small}
\begin{quote}
$\sym{\#define ID\_NAME 0}$\\
$\sym{\#define ID\_} \alpha_1 \sym{ 1}$\\
$\sym{   }\cdots$\\
$\sym{\#define ID\_} \alpha_n \sym{ }n$\\
$\sym{\#define MAX\_AGENTID} \sym{ } n$\\
\end{quote}
\end{small}
In addition, to manage symbols and arities, we define two arrays
$\sym{Symbols}$ and $\sym{Arities}$ as follows:
\begin{small}
\begin{quote}
$\sym{char Symbols[MAX\_AGENTID+1] = \{"", "} \alpha_1 \sym{",} \ldots \sym{,"}\alpha_n \sym{"\};}$\\
$\sym{int Arities[MAX\_AGENTID+1] = \{1,}  p_1 \sym{,} \ldots \sym{,} p_n \sym{\};}$
\end{quote}
\end{small}
\item $\sym{I=mkInterface(}n\sym{)}$.
This makes a global $n$-size array for the interface and corresponds to:
\begin{small}
\begin{program}
#define SIZE_INTERFACE \(n\)
Agent *I[SIZE_INTERFACE];
\end{program}
\end{small}
\item $x \, \sym{=} \sym{mkAgent(}\mathit{id}\sym{)}$. 
This makes a variable $x$ whose type is $\texttt{Agent}$ and 
assigns an agent node whose $\texttt{id}$ is $\mathit{id}$. 
This instruction corresponds to:
\sym{Agent *\(x\)=mkAgent(\(\mathit{id}\));}

\item $x \, \sym{=} \sym{mkName()}$.  This makes a variable $x$ whose
  type is $\texttt{Agent}$ and assigns an agent node whose \sym{id} is
  \sym{ID\_NAME}. Then it assigns \sym{NULL} to \sym{port[0]} of the
  $x$ in order to be distinguished from indirection nodes: \sym{Agent
    *\(x\)=mkAgent(ID\_NAME); \(x\)->port[0]=NULL;}

\item $\sym{free(}x\sym{)}$.
This disposes of a graph node assigned to $x$ (not recursively, just an assigned node):
\sym{freeAgent(\(x\));}

\item $x\sym{[}p\sym{]} \sym{=} y$.
This assigns a graph element $y$ to a port $p$ of an agent node $x$.
The port $p$ in LL0 corresponds to the port $p-1$ in the standardised 
implementation method, and thus this instruction 
corresponds to the following code:
\texttt{\(x\)[\(p-1\)]=\(y\);}

\item $x\sym{[0]} \sym{=} \alpha$.
This changes the id of an agent $x$ into $\alpha$. 
This corresponds to the following code:
\texttt{\(x\)->id=\(\alpha\);}

\item $\sym{push(}x\sym{,}y\sym{)}$.
This pushes two agents onto the equation stack.
This corresponds to the following code:
\texttt{pushActive(\(x\),\(y\));}

\item $\sym{stackFree()}$.
This disposes  of the top element of the equation stack.
In the translation result, it occurs in rule procedures.
In this implementation, 
the function $\texttt{popActive}$ manages the index of the equation stack,
and thus no code is required.
\end{itemize}

Next we manage the translated LL0 instructions for rule procedures.  A
rule procedure in LL0 such as ``$\sym{rule Alpha Beta}$'' is encoded
as a function that is named as $\sym{Alpha\_Beta}$, takes two pointers
$\sym{*a1}$ and $\sym{*a2}$ to two elements of the equation, and
creates nets according to interaction rules.  The special variables
$\sym{L}$ and $\sym{R}$ in the rule procedures are denoted as
$\sym{*a1}$ and $\sym{*a2}$, and thus
$\sym{L[1]},\sym{L[2]},\ldots,\sym{R[1]},\sym{R[2]},\ldots$ are
expressed as:
$\sym{a1->port[0]},\sym{a1->port[1]},\ldots,\sym{a2->port[0]},\sym{a2->port[1]},\ldots$.
As an example the rule procedures for $\sym{Add}$ and $\sym{Z}$ 
and for $\sym{Add}$ and $\sym{S}$ 
are encoded as follows:
\begin{multicols}{2}
\begin{small}
\begin{program}
 1. void Add_Z(Agent *a1, Agent *a2) \{
 2.   pushActive(a1->port[0],a1->port[1]);
 3.   freeAgent(a1);
 4.   freeAgent(a2);
 5. \}
\end{program}
\begin{program}
 1. void Add_S(Agent *a1, Agent *a2) \{
 2.   Agent *aS=mkAgent(ID_S);
 3.   Agent *aAdd=mkAgent(ID_Add);
 4.   Agent *w=mkName();
 5.   aAdd->port[0]=a1->port[0];
 6.   aAdd->port[1]=w;
 7.   pushActive(aAdd, a2->port[0]);
 8.   aS->port[0]=w;
 9.   pushActive(a1->port[1], aS);
10.   freeAgent(a1);
11.   freeAgent(a2);
12. \}
\end{program}
\end{small}
\end{multicols}

To manage these functions, we define a rule table $\sym{R}$, 
which stores pointers to those functions.
Here, for simplicity, we use the following simple matrix:
\begin{small}
\begin{program}
typedef void (*RuleFun)(Agent *a1, Agent *a2);
RuleFun R[MAX_AGENTID+1][MAX_AGENTID+1];
\end{program}
\end{small}
For instance, the above function is stored as:
\sym{R[ID\_Add][ID\_Z]=\&Add\_Z;}. The run-time function
$\sym{eval}$ is written as follows:
\begin{multicols}{2}
\begin{small}
\begin{program}
 1. void eval() \{
 2.  Agent *a1, *a2;
 3.  while (popActive(&a1, &a2)) \{
 4.   if (a2->id != ID_NAME) \{
 5.    if (a1->id != ID_NAME) \{ //Interact
 6.     R[a1->id][a2->id](a1, a2);
 7.    \} else if (a1->port[0] != NULL) \{
 8.     Agent *a1p0=a1->port[0]; //Ind1
 9.     freeAgent(a1);
10.     pushActive(a1p0, a2);
11.    \} else a1->port[0]=a2; //Var1
12.   \} else if (a2->port[0] != NULL) \{
13.     Agent *a2p0=a2->port[0]; //Ind2
14.     freeAgent(a2);
15.     pushActive(a1, a2p0);
16.   \} else a2->port[0]=a1; //Var2
17.  \}
18. \}
\end{program}
\end{small}
\end{multicols}

\section{Discussion}\label{sec:discussion}

To examine how data structures in INET, \inn, amineLight (the fastest
evaluator) and our implementation affect execution speeds, we
implemented a number of evaluators using the different encoding
methods.  We fix the number of ports as \sym{MAX\_PORT} that is
obtained during compilation, and we pre-populate the heap with these
nodes.  The fixed-size node representation has the disadvantage of
using more space than needed, but the advantage of being able to
manage and reuse nodes in a simpler way~\cite{Peyton}.  INET and \inn
are based on the graph calculus of interaction nets. Agent nodes are
represented as C structures:

\noindent\begin{minipage}[t]{9.3cm}
\begin{small}
\begin{program}
1. typedef struct Agent \{
2.   int id; struct Port *port[MAX_PORT];
3. \} Agent;  
4. typedef struct Port \{
5.   Agent *agent; int portNum;
6. \} Port;  
\end{program}
\end{small}
\end{minipage}
\begin{minipage}[t]{3cm}
(a) INET
\begin{center}
\includegraphics[scale=\smallscale,keepaspectratio,clip]{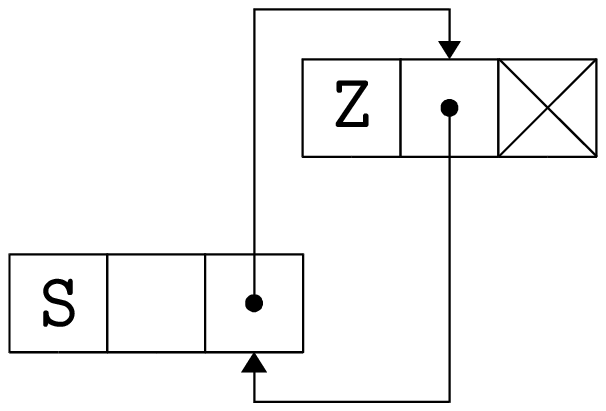}
\end{center}
\end{minipage}
\quad
\begin{minipage}[t]{3cm}
(b) \inn
\begin{center}
\includegraphics[scale=\smallscale,keepaspectratio,clip]{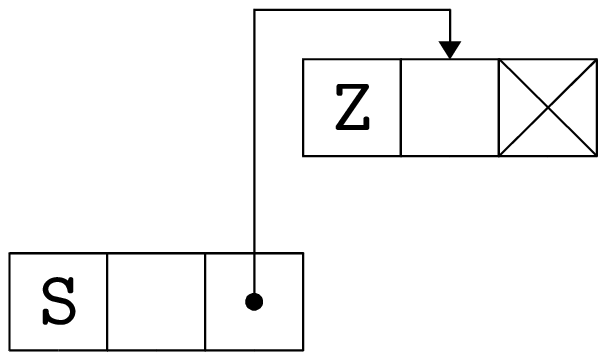}
\end{center}
\end{minipage}
\medskip

\noindent
In contrast with our method, the principal ports are
assigned to \sym{port[0]}, and connections between auxiliary ports are
encoded as mutual links between ports of agent nodes; in INET every
connection is linked mutually as shown in the above (a), while
\inn uses such mutual connections only for the
connection between auxiliary ports as shown in (b).  Although \inn has been proposed before INET, \inn
can be regarded as a refined version of INET.  We call the method to
use mutual links for the connection between auxiliary ports
\emph{undirected encoding}.

\begin{wrapfigure}[5]{R}{7.6cm}
\vspace{-4mm}
(c) Representation of $x=y$ in amineLight
\begin{center}
\includegraphics[scale=\smallscale,keepaspectratio,clip]{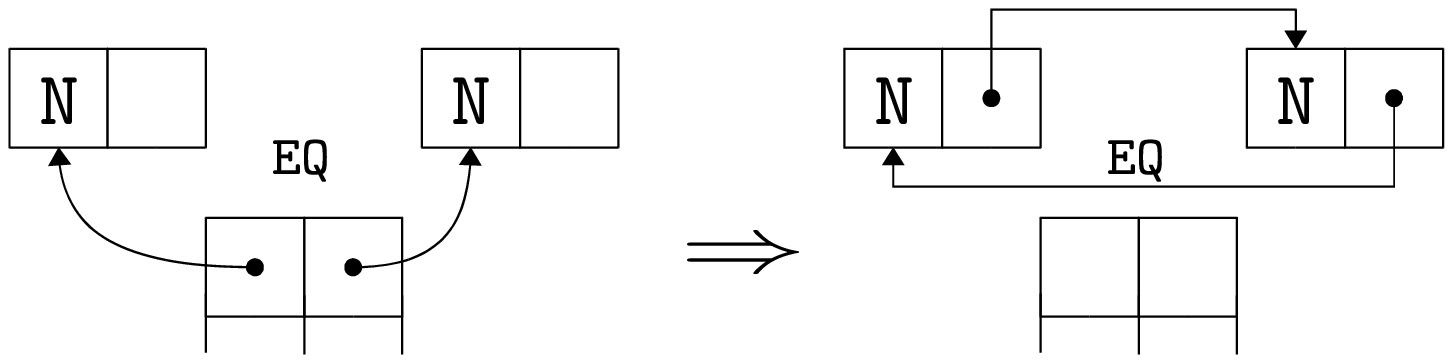}
\end{center}
\end{wrapfigure}
amineLight is based on the Lightweight calculus, and uses names to
represent connections between auxiliary ports, and every term is
encoded by single links at the start of the execution.  Our method,
called \emph{directed encoding}, uses single links for the connection
between auxiliary ports.  In the case of amineLight, an equation such
as $x=y$ becomes represented by mutual links during execution as shown
in the figure (c), while in our method, the equation preserves the
directed encoding, thus a single link
(Figure~\ref{fig:computation-rules-for-name-and-ind}).  In undirected
encoding method names do not occur, and in directed encoding method
substitution for each name is performed by removing the indirected
connection via the name locally.  Thus, the implementation needs no
environments for substitutions.  We do not garbage collect, taking
account of an optimisation mentioned subsequently in this section.

The table below shows execution times in seconds for computing
Fibonacci ($F_n$), Ackermann ($A$) and Church
numerals~\cite{MackieIC:yalyal}.  We see from the table that our
execution times are almost similar to those of amineLight and thus in
terms of the cost, the undirected encoding method of \inn is the best.

\begin{center}
{\small
\begin{tabular}{|r||r|r|r|r|}\hline
 &  \multicolumn{1}{r|}{Undirected(INET)} &  
\multicolumn{1}{r|}{Undirected(\inn)}
 & \multicolumn{1}{r|}{Directed(Light)}
 & \multicolumn{1}{r|}{Directed(Simpler)}   \\\hline\hline
$F_{32}$ & 1.58 & 1.37 & 1.52 & 1.49 \\\hline
$F_{33}$ & 2.62 & 2.29 & 2.52 & 2.49 \\\hline
$F_{34}$ & 4.37 & 3.80 & 4.21 & 4.15\\\hline
$A(3,10)$ & 1.77 & 1.42 & 1.59 & 1.58\\\hline
$A(3,11)$ & 7.12 & 5.73 & 6.44 & 6.39 \\\hline
$A(3,12)$ & 29.47 & 24.01 & 26.39 & 26.14 \\\hline
\texttt{2 7 6 I I} & 0.73 & 0.71 & 1.26 & 1.28 \\\hline
\texttt{2 7 7 I I} & 2.12 & 2.13 & 3.58 & 3.68\\\hline
\end{tabular}
}
\end{center}

In comparison to amineLight, our implementation computes 
Fibonacci numbers and Ackermann function a little faster.
On the other hand, amineLight performs better in the 
Application of Church numerals. The reason is that Application of Church numerals 
demand a lot of computation for names,
especially for equations such as $x=y$,  
yet these operations 
require extra computational steps in our implementation.
To illustrate this point,  let us look at the computation of the following sequence of equations:
$\alpha=x, \, y=\beta, \, x=y$.
The lightweight abstract machine in amineLight reduces 
it to $\beta=\alpha$ in two steps:
$\conf{}{\alpha=x, \, y=\beta, \, x=y} \rtoCom
\conf{}{\alpha=y, \, y=\beta} \rtoCom
\conf{}{\alpha=\beta}$,
whereas our encoding method takes four steps:
$\confSimple{}{\alpha=x, \, y=\beta, \, x=y} \lra 
\confSimple{}{\alpha=\ind{y}, \, y=\beta} \lra 
\confSimple{}{\alpha=y, \, y=\beta} \lra 
\confSimple{}{\alpha=\ind{\beta}} \lra 
\confSimple{}{\alpha=\beta}$.
This is because Lightweight calculus manages both sides of an equation,
while Simpler one manages only a single side.
To illustrate further, the table below
shows ratios of name operations (denoted as ``N'') to interaction operations (as ``I'').
\begin{center}
{\small
\begin{tabular}{|r|r||r|r||r|r|}\hline
 & \multicolumn{1}{c||}{\multirow{2}{*}{I}} & \multicolumn{2}{c||}{Light} & \multicolumn{2}{c|}{Simpler}\\\cline{3-6}
 &   &  \multicolumn{1}{c|}{N} & N/I 
 &  \multicolumn{1}{c|}{N} & N/I\\\hline\hline
$F_{32}$ & 74636718 & 51008017 &  0.68 &  65106325 &  0.87 \\\hline
$F_{33}$ & 123315177 & 82532797 & 0.67 &  105344341 & 0.85 \\\hline
$F_{34}$ & 203654818 & 133540964 & 0.66 &  170450820 & 0.84 \\\hline
$A(3,10)$ & 134103148 & 134094952 & 1.00 &  134094952 & 1.00 \\\hline
$A(3,11)$ & 536641652 & 536625264 & 1.00 &  536625264 & 1.00 \\\hline
$A(3,12)$ & 2147025020 & 2146992248 & 1.00 &  2146992248 & 1.00 \\\hline
\texttt{2 7 6 I I} & 15676873 & 43111255 & 2.75 &  64538288 & 4.12 \\\hline
\texttt{2 7 7 I I} &  46118916& 126826871 & 2.75 &   190190039 & 4.12 \\\hline
\end{tabular}
}
\end{center}
With respect to the computation of Application of Church numerals,
the ratio increases to 4.12 compared to only 2.75 in the amineLight encoding method.
Even though the cost of each operation for those names is quite small, 
as shown in the computation of Fibonacci number 
that is faster where the ratio increases by 0.18,
the much accumulation of the cost induces the less efficiency. We anticipate that it
is possible to reduce the cost of operations for names by enhancing the data-structures.

\begin{wrapfigure}[7]{R}{7.5cm}
\vspace{-4mm}
\includegraphics[scale=\tinyscale,keepaspectratio,clip]{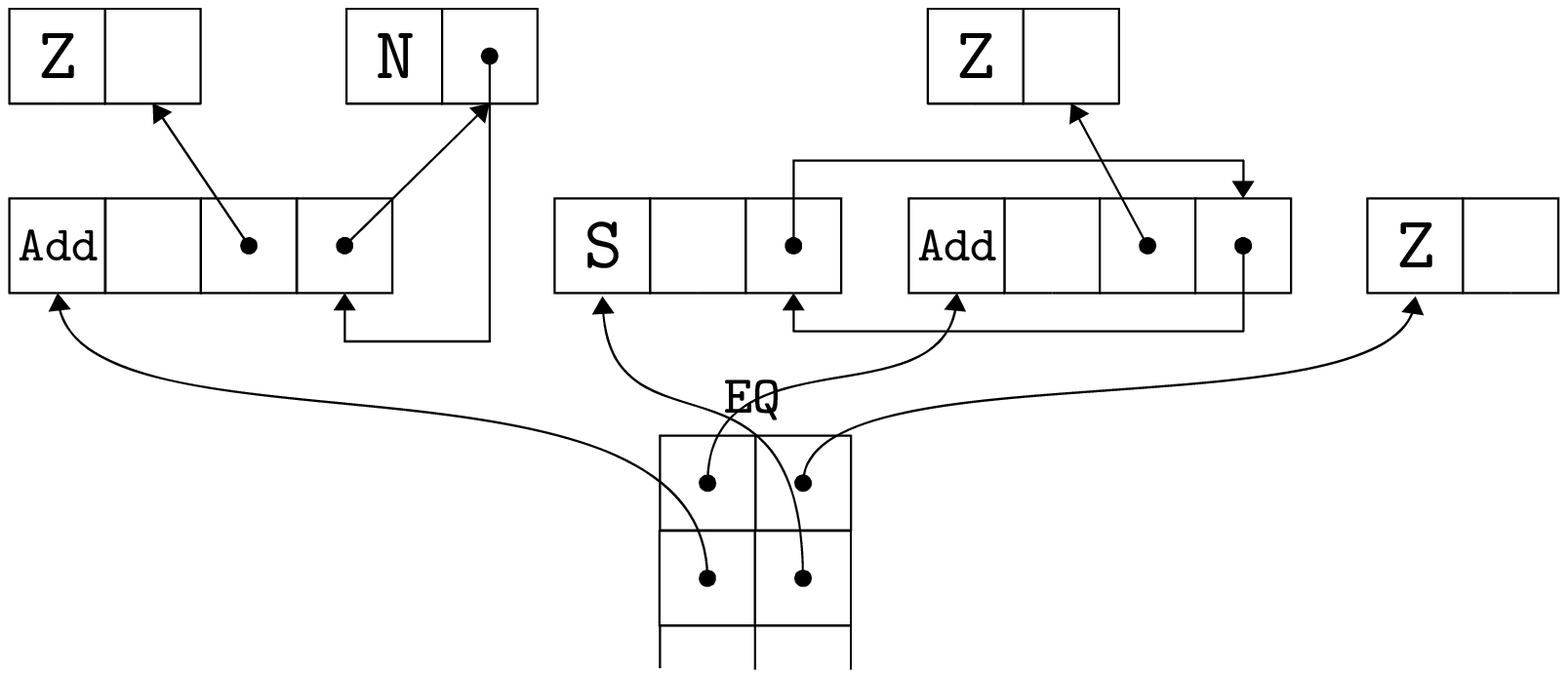}
\end{wrapfigure}
The advantage of the directed encoding method is locality of the
rewriting in parallel execution.  An active pair must be reduced with
the interface preserved, and thus reduction of two active pairs that
are connected via an auxiliary port(s) of an interacting agent need to
be managed differently because each rewrite will update the same set
of auxiliary ports.  As an example take the graph shown in the right
side figure, which is a graph encoded in \inn of the net in
Figure~\ref{example_connection_auxs_node2}.  The two active pairs are
connected to each other via the auxiliary port of the interacting
agent $\sym{Add}$ and $\sym{Z}$, and this connection information must
be preserved when the active pairs are reduced at once. This checking
process could be spread into other parts of the net globally.  In the
case of the directed encoding method, the connection is preserved by a
name as shown in Figure~\ref{example_connection_auxs_node2}, and thus
reduction of the two active pairs are performed in parallel as long as
critical sections are used to manage names.

\begin{wrapfigure}[6]{R}{3.8cm}
\vspace{-3mm}
\includegraphics[scale=\tinyscale,keepaspectratio,clip]{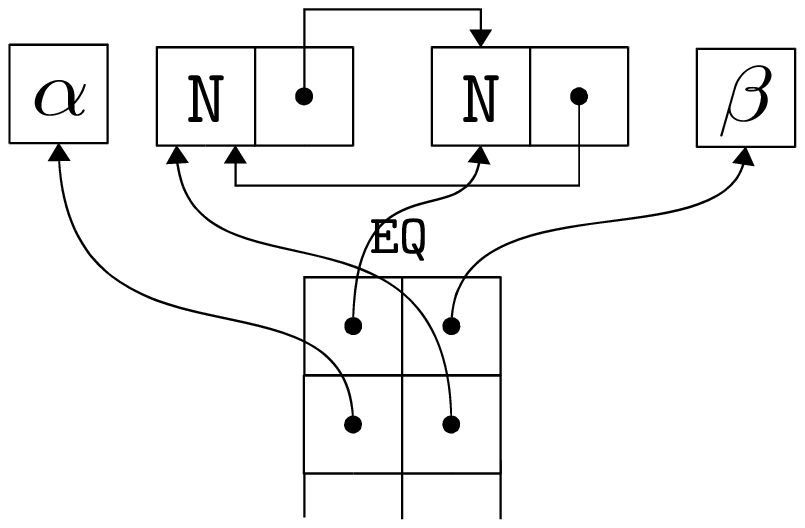}
\end{wrapfigure}
The mutual links affect the locality and thus we have proposed the new
method of encoding so that a connection between names 
can be represented by a single link.
With respect to the encoding method in amineLight, though it is the directed one, 
the connections between names are represented 
as mutual links (Figure (c)) and we need to check for the lock 
and this can also spread globally.
Take the right side figure as an example. This shows a graph after the first step computation of $\alpha=x, \, y=\beta, \, x=y$.
The two elements of the stack should not be performed at once because each rewriting affects another, so the checking process is also required.

In addition, our model is simpler than the model of amineLight 
in terms of dealing with equations, 
thus only a single side of an equation is managed.
This derives less critical sections that are caused only by the computational rules Var1 (Figure~\ref{fig:computation-rules-for-name-and-ind} (a)) and Var2,
since name nodes can be pointed-to by two active pairs (that is, auxiliary ports of an active pair are connected).
Moreover, those are performed by connecting the ports of names to
other principal ports of unlocked agent nodes, therefore
these can be locked with an atomic operation such as Compare-and-swapping as follows:
\begin{small}
\begin{multicols}{2}
\begin{program}
 1. void eval() \{
 2.  Agent *a1,*a2;
 3.  while (popActive(&a1,&a2)) \{
 4. loop:
               \vdots
12.    \} else if (!(__sync_bool_compa
   re_and_swap(&(a1->port[0]),NULL,a2))) 
13.     goto loop; //retry
\end{program}
\end{multicols}
\end{small}

We finish this section by outlining an important optimisation, but
leave the implementation details for future work.  Once a net is
compiled into an instruction list of LL0, operations such as
producing, disposing and connecting ports of agents is done at the
level of execution of those instructions.
We illustrate the optimisation by considering the rule between $\sym{Add}$ and $\sym{S}$:
$\sym{Add}(x_1,x_2) = \sym{S}(y)\ito \sym{Add}(x_1,w)=y,  x_2=\sym{S}(w)$.
\noindent
The compilation result of this rule illustrated in
Example~\ref{Example:compilation-Add-SZ}.  In the right-hand side of
this rule, the active pair agents $\sym{Add}$ and $\sym{S}$ also
occur.  Thus, instead of producing new agents, it is possible to reuse
the active pair agents.  By introducing \sym{StackL} and \sym{StackR}
to refer the top elements of the stack $\mathit{EQ}$, it is possible
to obtain an alternative sequence of instructions where number of
instructions, especially for heap allocations, decreases and thus
faster execution is expected:
\begin{multicols}{3}
\begin{small}
\begin{program}
1. rule Add S \{
2.   w=mkName()
3.   x2=StackL[2]
4.   tmpR=StackR
5.   StackR=tmpR[1]
6.   tmpR[1]=w
7.   push(x2,tmpR)
8. \}
\end{program}
\end{small}
\end{multicols}

\section{Conclusion}\label{sec:conc}

In this paper we have designed a simple data-structure for
representing interaction nets, and designed a corresponding calculus
that has a direct relationship with the structure. As a consequence,
we can use the calculus to reason about the rewriting process, and
also to study the cost of reduction. This led to an investigation into
optimising rules which we outlined in Section 5. We believe that this
model can provide the basis for further development implementation
technology for interaction nets.

\bibliography{bibfile}

\end{document}